\documentclass{article}
\usepackage[utf8]{inputenc}

\usepackage{epsfig}
\usepackage{amsfonts}
\usepackage{amssymb}
\usepackage{amstext}
\usepackage{amsmath}
\usepackage{amsthm}
\usepackage{xspace}

\usepackage{tabularx}
\usepackage{booktabs}
\usepackage{multirow}
\usepackage{hyperref}
\usepackage{cleveref}
\usepackage{url}
\usepackage{color}

\newtheorem{thm}{Theorem}
\newtheorem{lemma}[thm]{Lemma}
\newtheorem{definition}[thm]{Definition}
\newtheorem{corollary}[thm]{Corollary}

\newtheorem{problem}[thm]{Problem}
\newtheorem{conj}[thm]{Conjecture}

\title{Inapproximability of Additive Weak Contraction under \textsc{SSEH} and \textsc{Strong UGC}}
\author{Siddhartha Jain\thanks{IIIT-Delhi, India.  Email: \url{siddhartha16269@iiitd.ac.in}}}

\begin{document}

\maketitle
\begin{abstract}
Succinct representations of a graph have been objects of central study in computer science for decades. In this paper, we study the operation called \emph{Distance Preserving Graph Contractions}, which was introduced by Bernstein et al. 
(ITCS, 2018). This operation gives a minor as a succinct representation of a graph that preserves all the distances of the original (up to some factor). The graph minor given from contractions can be seen as a dual of spanners as the distances can only shrink (while distances are stretched in the case of spanners). Bernstein et al. proved inapproximability results for the problems of finding maximum subset of edges that yields distance preserving graph contractions for almost major classes of graphs except for that of Additive Weak Contraction. 

The main result in this paper is filling the gap in the paper of Bernstein et al. We show that the Maximum Additive Weak Contraction problem on a graph with $n$ vertices is inapproximable up to a factor of $n^{1-\epsilon}$ for every constant $\epsilon>0$. Our hardness results follow from that of the Maximum Edge Biclique (\textsc{MEB}) problem whose inapproximability of $n^{1-\epsilon}$ has been recently shown by Manurangsi (ICALP, 2017) under the \textsc{Small Set Expansion Hypothesis (SSEH)} and by Bhangale et al. (APPROX, 2016) under the \textsc{Strong Unique Games Conjecture (SUGC)} (both results also assume $\mathrm{NP}\not\subseteq\mathrm{BPP}$).
\end{abstract}

\section{Introduction}
Coping with very large networks has been a major issue in computer science in both theory and practice. 
Operations acting on extremely large networks are rather time-consuming and not quite satisfactory for most applications.
The techniques of graph compression have been posited as a solution,  as they reduce the size of the network and thus reduce the processing time.
While compressing it is impossible to retain all the information about the original network. So we aim to retain specific properties of the input. This motivates the study of an object called {\em Spanner} that keeps only a sparse subset of the edges of a network, while preserving the distances between every pair of vertices up to some multiplicative (resp., additive) factor. To be formal, in graph theory terminology, a $k$-spanner $H$ of a graph $G$ is a sparse subgraph of $G$ such that every pair of vertices in $H$ has distance at most a multiplicative (resp., additive) factor $k$ away from the original distance in $G$. 

Bernstein et al. \cite{bernstein_et_al:LIPIcs:2018:8342} recently proposed a compression operation called {\em Distance Preserving Graph Contraction}. Here we obtain a subgraph by contracting subsets of edges, while promising a lower bound on the distances, which results in a minor of the input graph. Observe that the contraction operation can only decrease the distance of each pair of vertices, whereas deleting edges as in the case of spanners can only increase the distances. The contracted graph can be thought of as a dual of spanners.

In \cite{bernstein_et_al:LIPIcs:2018:8342}, the authors introduced the problem of finding the maximum number of edges whose contraction produces a minor that guarantees the distance between any two vertices to be shorter by a factor of at most $k$, namely, the {\em $k$-Contraction} problem. The authors also defined the relaxed variant, namely the {\em Weak $k$-Contraction} problem.
Note that, the Contraction problem is studied only in the case of additive contraction; this is because, in the multiplicative case, no edges can be contracted. In each variant, Bernstein et al. studied several important graph classes and either present polynomial-time algorithm or show NP-hardness results for the problems; see \Cref{tab:known-results}.
We list their results with ours.

\begin{table}
    \begin{tabular}{ |p{3cm}|p{3cm}|p{5cm}|  }
    \hline
    \multicolumn{3}{|c|}{Inapproximability results} \\
    \hline
    Tolerance & Type & Hardness \\
    \hline
    Additive ($\alpha=1$) & Contraction & $m^{\frac{1}{2}-\epsilon}$-inapx on bipartite graphs, $l=1$ \cite{bernstein_et_al:LIPIcs:2018:8342}\\
    Additive ($\alpha=1$) & Contraction & $n^{1-\epsilon}$-inapx \cite{bernstein_et_al:LIPIcs:2018:8342}\\
    Additive ($\alpha=1$) & Weak Contraction & $n^{1-\epsilon}$-inapx on bipartite graphs, $l=1$ [\bf{Thm \ref{thm:hard}}]\\
    Multiplicative ($\beta=0$) & Weak Contraction & $n^{1-\epsilon}$-inapx \cite{bernstein_et_al:LIPIcs:2018:8342} \\
    \hline
    \end{tabular}
    \caption{Known Results on the (Weak) Contraction problems.}
    \label{tab:known-results}
\end{table}

\section{Preliminaries}

\subsection{Notations}

We use $G(V,E,d)$ or simply $G$ to refer to a simple, undirected graph with a non-negative distance function $d : E \rightarrow \mathbb{R}^{+}$. Throughout, we assume that $G$ has $n$ veritces and $m$ edges. We assume further that the graph $G$ is connected since, otherwise, we can solve the subproblem on each connected component separately. 
Given a set of contracted edges $C \subseteq E$ (we will mostly use $C$ to denote contracted edges), the distance function induced by $d$ after the contraction is denoted by $d_C$.

\subsection{Distance-Preserving Contractions}

\begin{definition}
Given a graph $G(V, E)$ and a distance function $d: E \rightarrow \mathbb{R}^+$, an {\em $(\alpha,\beta)$-contraction} of $G$ is a set of edges $C\subseteq E$ such that $d_C(v, u) \geq d(v, u)/\alpha - \beta$ for all vertices $v, u \in V$. We abbreviate $(\alpha,\beta)$-contraction by  \textsc{Cont}$(\alpha, \beta)$.
\end{definition}

The concept of Weak Contraction is defined to allow enough flexibility to make Multiplicative (Weak) Contraction non-trivial.

\begin{definition}
Given a graph $G(V, E)$ and a distance function $d: E \rightarrow \mathbb{R}^+$, an {\em $(\alpha,\beta)$-weak-contraction} of $G$ is a set of edges $C\subseteq E$ such that the following two properties hold:
\begin{itemize}
    \item $C \subsetneq E$
    \item $d_C(v, u) \geq d(v, u)/\alpha - \beta$ whenever $d_C(v, u) \neq 0$ for all $v, u \in V$
\end{itemize}
We abbreviate $(\alpha,\beta)$-weak-contraction by \textsc{WeakCont}$(\alpha, \beta)$.
\end{definition}

% {\color{red} NEED TO REWRITE: The problem we want to solve is that we want to find the set of contracted edges $C$ which maximises $|C|$ for a given $G$. We use the notation to refer to the problem and the set interchangeably. }

We can now formulate the corresponding problems for these structures.

\begin{problem}
Given a graph $G(V, E)$ and a distance function $d: E \rightarrow \mathbb{R}^+$, find an optimal \textsc{Cont}$(\alpha,\beta)$, $C^{*}$ of $G$. $C^{*}$ is said to be optimal if there does not exist a \textsc{Cont}$(\alpha, \beta)$, $C$ of $G$ such that $|C| > |C^{*}|$.
\end{problem}

\begin{problem}
Given a graph $G(V, E)$ and a distance function $d: E \rightarrow \mathbb{R}^+$, find an optimal \textsc{WeakCont}$(\alpha,\beta)$, $C^{*}$ of $G$. $C^{*}$ is said to be optimal if there does not exist a \textsc{WeakCont}$(\alpha, \beta)$, $C$ of $G$ such that $|C| > |C^{*}|$.
\end{problem}

\subsection{Complexity Assumptions}

We now discuss the Complexity Theoretic assumptions that are required for the result by each of the previous works. For completeness, we reproduce the definitions here.

\paragraph*{The Strong Unique Games Conjecture (SUGC).} This conjecture was first introduced by Bansal and Khot \cite{Bansal:2009:OLC:1747597.1748061}, but coined by Bhangale et al. \cite{bhangale_et_al:LIPIcs:2016:6272} with a slight modification.

\begin{definition}
Given a bi-regular bipartite graph $G(U, V, E, L, \{\pi_e\}_{e\in E})$ and a labelling $l:U\cup V \rightarrow [L]$, we say an edge $e=(u, v)$ is satisfied if $\pi_e(l(v)) = l(u)$. Moreover, let $s_l$ be the fraction of edges satisfied by the labelling $l$.
\end{definition}

\begin{conj} \cite{bhangale_et_al:LIPIcs:2016:6272}
For all $\delta, \eta, \gamma > 0$ there exists a $L \in \mathbb{N}$ such that given $G(U, V, E, L, \{\pi_e\}_{e\in E})$, it is \textsc{NP-Hard} to distinguish between the following two cases:

\begin{itemize}
    \item There exists an $l$ such that $s_l \geq 1 - \eta$
    \item For all $l$, $s_l\leq \gamma$. Moreover, for all $S\subseteq V$ such that $|S| = \delta|V|$, we have that $|\Gamma(S)| \geq (1-\delta)|U|$ where $\Gamma(S) = \{u \in U | \exists v ~s.t.~ (u, v) \in E\}$
\end{itemize}
\end{conj}

\paragraph*{The Small Set Expansion Hypothesis (SSEH).} This conjecture was introduced by Raghavendra and Steurer \cite{Raghavendra:2010:GEU:1806689.1806792} to overcome the shortcomings of UGC while trying to prove hardness results, and it has received immense attention since. It is in fact equivalent to a stronger version of \textsc{UGC} but distinct from \textsc{Strong UGC}.

\begin{definition}
Given a $d$-regular graph $G(V, E)$, for every $S\subseteq V$ let us define 
$$\phi_G(S) = \frac{|E(S, V\setminus S)|}{d|S|}$$
Moreover, for every $\delta \in [0, 1/2]$ we define 
$$\phi_G(\delta) = \min_{|S| = \delta |V|} \phi_G(S)$$
\end{definition}

\begin{conj}
\cite{Raghavendra:2010:GEU:1806689.1806792} For every $\eta > 0$ there exists a $\delta$ such that given a graph $G(V, E)$ it is \textsc{NP-Hard} to distinguish between the following two cases:
\begin{itemize}
    \item $\phi_G(\delta) \geq 1 - \eta$
    \item $\phi_G(\delta) \leq \eta$
\end{itemize}
\end{conj}

We reduce the Maximum Edge Biclique (\textsc{MEB}) problem to finding the largest \textsc{WeakCont}$(1, \beta)$ as defined in the paper by Bernstein et al. \cite{bernstein_et_al:LIPIcs:2018:8342}, called Weak Contraction with tolerance function $\phi(x) = x - \beta$.

\textsc{MEB} was proven to be hard to approximate under \textsc{SSEH} by Manurangsi \cite{manurangsi:LIPIcs:2017:7500}. It was also proven to be hard to approximate under the \textsc{Strong UGC} by Bhangale et al \cite{bhangale_et_al:LIPIcs:2016:6272}.

\subsection{Biclique} 

\begin{problem}
Maximum Edge Biclique (\textsc{MEB}): given a bipartite graph G, find a complete bipartite subgraph of G with maximum number of edges.
\end{problem}

\begin{problem}
Maximum Balanced Biclique (MBB): given a bipartite graph G, find a balanced complete bipartite subgraph of G with maximum number of vertices.
\end{problem}

\begin{thm}
\label{thm:pasin}\cite{manurangsi:LIPIcs:2017:7500} Assuming SSEH, there is no polynomial time algorithm that approximates
\textsc{MEB} or \textsc{MBB} to within $n^{1-\epsilon}$ factor of the optimum for every $\epsilon > 0$, unless $\textsc{NP} \subseteq \textsc{BPP}$.
\end{thm}

Concretely, they prove this by showing the following lemma.

\begin{lemma}
\label{lem:pasin}\cite{manurangsi:LIPIcs:2017:7500} Assume SSEH. Then given a bipartite graph $G = (L \cup R, E)$ with $|L| = |R| = n$, for every $\delta > 0$ it is \textsc{NP-hard} to distinguish between the following two cases:
\begin{itemize}
    \item (Completeness) G contains $K_{(\frac{1}{2} - \delta)n,(\frac{1}{2} - \delta)n}$ as a subgraph.
    \item (Soundness) G does not contain $K_{\delta n,\delta n}$ as a subgraph.
Here $K_{t,t}$ denotes the complete bipartite graph in which each side contains t vertices.
\end{itemize}
\end{lemma}
This lemma works for \textsc{MBB} but also implies a similar result for \textsc{MEB}. Since the paper does not explicitly state this reduction, we state it here for exposition.

\begin{lemma}
Assume SSEH. Then given a bipartite graph $G = (L \cup R, E)$ with $|L| = |R| = n$, for every $\delta > 0$ it is \textsc{NP-hard} to distinguish between the following two cases:
\begin{itemize}
    \item (Completeness) G contains a biclique having $((\frac{1}{2} - \delta)n)^2$ edges..
    \item (Soundness) G does not any biclique containing $(\delta n)^2$ edges.
\end{itemize}
\end{lemma}

\begin{proof}
Given input graph $G = (L \cup R, E)$ for \textsc{MEB}, we construct $G \times G$ (tensor product). Let us divide the vertex set of $G\times G$ into 4 parts. $V(G\times G) = (L, L) \cup (L, R) \cup (R, L) \cup (R, R)$. The only edges in this graph are between $(L, R)$ and $(R, L)$ and hence this graph is also bipartite, of size $\Theta(n^2)$. Let $N = n^2$.

\begin{itemize}
    \item (Completeness) A $K_{n_1, n_2}$ in $G$ corresponds to a $K_{n_1 n_2, n_1 n_2}$ in $G\times G$. Therefore, any biclique in $G$ with $t = n_1 n_2$ edges corresponds to a $K_{t, t}$ in $G\times G$. Let $t = ((\frac{1}{2} - \delta)n)^2$.
    
    \item (Soundness) Given a $K_{t, t}$ in $G\times G$, we know that for all $(v_l, v_r) \in V(K_{t,t})$ and $(u_l, u_r) \in V(K_{t,t})$, by definition $(v_l, u_r), (v_r, u_l) \in E$. Hence we have a biclique with $t^2$ edges in $G$. If there is no $K_{t, t}$ in $G\times G$ for $t = \delta n$, there can be no biclique of size $(\delta n)^2 = \delta^2 N$ in G.
\end{itemize}
\end{proof}

Theorem \ref{thm:pasin} follows from gap amplification via randomized graph product which is discussed in Appendix B of the full version of \cite{manurangsi:LIPIcs:2017:7500}. This is also where we assume \textsc{NP} $\not\subseteq$ \textsc{BPP}. Without this assumption (assuming only \textsc{P}$\neq$\textsc{NP}), both \cite{manurangsi:LIPIcs:2017:7500} and \cite{bhangale_et_al:LIPIcs:2016:6272} show that the problem cannot be approximated within any constant factor, under \textsc{SSEH} and \textsc{Strong UGC} respectively.

\section{Reduction}
First let us see a simple property of \textsc{WeakCont}$(1, 1)$ on graphs with unit edge lengths.

\begin{lemma}
\label{lem:path}
For any path $P$ in $G$, if two disjoint edges $(u_1, v_1), (u_2, v_2) \in P \cap C$ then $P \subset C$.
\end{lemma}

\begin{proof}
By way of contradiction let $(s, t) \in P \setminus C$. One of $d(s, v_2), d(t, u_1), d(u_1, v_2)$ reduces by 2 after contraction, depending on the relative position of $(s, t)$.
\end{proof}

Our results are based on a simple reduction gadget, which we describe first.

Given an input bipartite graph $G=(V \cup U, E)$ we create an instance $B_G$ as follows. We create two copies each of $V$ and $U$, let us call them $V_a, V_b, U_a, U_b$. The subset of vertices $V_a, U_a$ induces $G$. We also add edges of the form $(v_a, v_b)$ and $(u_a, u_b)$ for all $v \in V, u \in U$ to $E(B_G)$. An illustration in Fig \ref{fig:gadget}.

\begin{figure}[h]
    \centering
    \includegraphics{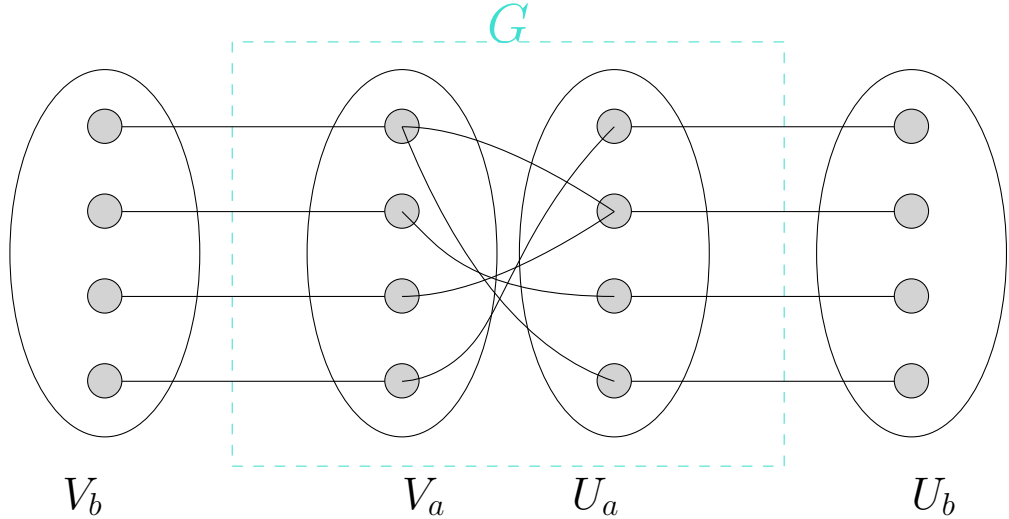}
    \caption{A representative drawing of $B_G$}
    \label{fig:gadget}
\end{figure}

The following lemma will help us prove that finding the largest \textsc{WeakCont}$(1,1)$ is hard.

\begin{lemma}
\label{lem:biclique}
Assigning weights 1 to all edges in $B_G$, any $(1,1)$ Weak Contraction of size strictly greater than 1 must contract a biclique in $B_G$.
\end{lemma}
\begin{proof}
% There are two kinds of edges we can contract. One are the ``matching"" edges of the form $(v_a, v_b)$ or $(u_a, u_b)$. The other edges correspond to those in $G$.

% If we contract any edge of the form $(v_a, v_b)$ then we will show we cannot contract any other edge in the graph. So assume by way of contradiction, that another ``matching" edge was contracted. Let the vertex from this edge in $V_b$ be $v_{b}'$. Since $d(v_b, v_b') \geq 3$, the only way 

% So now we can focus on the bipartite graph induced by $V_a \cup U_a$. Let $C$ be the set of edges contracted. Assume by way of contradiction that there are 4 vertices involved in $C$ which do not induce a complete $K_{2,2}$ in $B_G$. Let $v_a, u_a$ be 2 vertices missing an edge in this scenario. Consider $d(v_b, u_b)$. Since we cannot decrease this by more than 1, we can only contract 1 edge in this sub-graph. Moreover, if there is any other edge contracted in the graph, say $(v_a', u_a')$, then $d(v_b, u_b')$ reduces by 2, which is a contradiction.

% We also claim that you can only contract atmost 1 biclique. This is because if you contract multiple bicliques, take $v_a$ belonging to the first clique and $u_a$ belonging to the second clique. Then $d(v_b, u_b)$ reduces by 2, which is a contradiction.

There are two kinds of edges we can contract. One are the ``matching" edges of the form $(v_a, v_b)$ or $(u_a, u_b)$. The other edges correspond to those in $G$.

First, we will see that if a ``matching" edge belongs to $C$, it must be that $|C| = 1$, or in other words that it is the only edge contracted. Without loss of generality let contracted edge be $(v_a, v_b)$.

\begin{itemize}
    \item First let us show one cannot contract any disjoint edge in the graph induced by $V_a \cup U_a$. By way of contradiction $(s, t)$ is contracted. Any other edge in the graph can be expressed as a part of a path including these edges. Hence by Lemma \ref{lem:path}, the entire graph must be contracted. But that means $C$ is not a \textsc{WeakCont}$(1,1)$.
    
    \item Now say some edge $(v_a, u_a)$ was contracted. Then $(u_a, u_b)$ must also be contracted, else $d(v_b, u_b)$ decreases by 2. Let us take some other edge $(s, t)$ induced by $V_a \cup U_a$. There is a path $P_1$ from $v_a$ to $t$. Similarly $P_2$ from $u_a$ to $s$. Now the path $\{(v_b, v_a)\} \cup P_1 \cup \{(s, t)\} \cup P_2 \cup \{(u_a, u_b)\}$ satisfies the condition for Lemma \ref{lem:path}. Hence, $(s, t) \in C$. But we have seen that once such an edge is contracted it leads to a contradiction.
    
    \item Now the only edges left to contract are the matching edges themselves. But if we contract any other matching edge, we can take a path including this edge and $(v_b, v_a)$. Now we are forced to contract an edge in $V_a \cup U_a$ which we have seen leads to a contradiction previously.
\end{itemize}

So now we can focus on the graph induced by $V_a \cup U_a$. Let $v_b, u_b$ be vertices such that their neighbours $v_a, u_a$ are endpoints of some edge in $C$. Let these edges be $(v_a, u_a')$ and $(v_a', u_a)$. Let us assume they are not involved in a biclique, which means $(v_a, u_a)$ does not exist. Since the graph is connected, there exists a path $P$ from $v_a'$ to $u_a'$. Let the path from $v_b$ to $u_a'$ be $P_1$ and the path from $u_b$ to $v_a'$ be $P_2$. Let us consider $P_1 \cup P \cup P_2$. Since we have contracted more than one edge in this path, we must contract the entire path. But since we know contracting matching edges leads to a contradiction, this is not possible. Therefore we must contract a biclique.
\end{proof}

\begin{thm}
\label{thm:hard}
Assuming SSEH \textsc{OR} assuming \textsc{Strong UGC}, no polynomial algorithm can approximate \textsc{WeakCont}$(1,1)$ within a factor of $n^{1-\epsilon}$ for every $\epsilon > 0$, even for bipartite graphs with unit edge lengths, unless \textsc{NP} $\subseteq$ \textsc{BPP}.
\end{thm}

\begin{proof}
For every $\frac{1}{4} > \delta > 0$
\begin{itemize}
    \item (Completeness) If we have a biclique with $((\frac{1}{2}-\delta)n)^2$ edges in $G$, since $V_a \cup U_a$ induce $G$ in $B_G$, we have a biclique of the same size in $B_G$.
    
    \item (Soundness) If we have no biclique of size $(\delta n)^2$ in $G$, there is no biclique of that size in $B_G$. By Lemma \ref{lem:biclique}, the largest $(1,1)$ Weak Contraction contracts less than $(\delta n)^2$ edges.
\end{itemize}
\end{proof}

\begin{corollary}
Assuming SSEH \textsc{OR} assuming \textsc{Strong UGC}, no polynomial algorithm can approximate \textsc{WeakCont}$(1,\beta)$ within a factor of $n^{1-\epsilon}$ for every $\epsilon > 0$, even for bipartite graphs with unit edge lengths, unless \textsc{NP} $\subseteq$ \textsc{BPP}.
\end{corollary}
\begin{proof}
In our construction if we replace the edge weights with $\beta$, everything follows similarly.
\end{proof}

\section{Acknowledgements}
The author thanks Bundit Laekhanukit for his mentorship and useful discussion on the problem, especially for suggesting the inapproximability results for the Maximum Biclique problem. In addition, the author would like to thank him for funding through his 1000-talents award by the Chinese government.
This work was partially done while the author was visiting the Institute for Theoretical Computer Science at the Shanghai University of Finance and Economics (ITCS@SUFE), and the author would like to thank ITCS@SUFE for the office space and all the support.

% \newpage
\newcommand{\etalchar}[1]{$^{#1}$}

\end{document}